\documentclass[sn-mathphys]{sn-jnl}
\usepackage[english]{babel}
\usepackage[T1]{fontenc}

\usepackage{amssymb}
\usepackage{latexsym}
\usepackage{enumerate}
\ifx\smallsetminus\undefined
\def\smallsetminus{\setminus}
\fi

\newcommand\cB{\mathcal B}
\newcommand\cE{\mathcal E}

\newcommand\cP{\mathcal P}

\newcommand\cF{\mathcal F}

\newcommand\cH{\mathcal H}

\newcommand{\cI}{\mathcal I}

\newcommand{\fA}{\mathfrak{A}}

\newcommand{\fS}{\mathfrak{S}}
\newcommand{\PG}{\mathrm{PG}}
\newcommand{\AG}{\mathrm{AG}}
\newcommand{\GF}{\mathrm{GF}}

\newcommand{\Tr}{\mathrm{Tr}\,}
\def\cC{{\mathcal C}}
\def\cV{{\mathcal V}}

\jyear{2022}

\theoremstyle{thmstyleone}
\newtheorem{theorem}{Theorem}
\newtheorem{proposition}[theorem]{Proposition}

\theoremstyle{thmstyletwo}
\newtheorem{example}{Example}
\newtheorem{remark}{Remark}

\theoremstyle{thmstylethree}
\newtheorem{definition}{Definition}

\begin{document}
%\nocite{label}

\title[Hypersurfaces,  minimal codes and secret sharing schemes]{Some hypersurfaces over finite fields,  minimal codes and secret sharing schemes}

\author*[1]{\fnm{Angela} \sur{Aguglia}}\email{angela.aguglia@poliba.it}
\equalcont{All authors contributed equally to this work.}

\author[1]{\fnm{Michela} \sur{Ceria}}\email{michela.ceria@poliba.it}
\equalcont{All authors contributed equally to this work.}

\author[2]{\fnm{Luca} \sur{Giuzzi}}\email{luca.giuzzi@unibs.it}
\equalcont{All authors contributed equally to this work.}

\affil*[1]{\orgdiv{Dipartimento di Meccanica, Matematica e Management}, \orgname{Politecnico di Bari}, \orgaddress{\street{Via Orabona 4}, \city{Bari}, \postcode{I-70125},  \country{Italy} }}

\affil[2]{\orgdiv{DICATAM}, \orgname{University of Brescia}, \orgaddress{\street{Via Branze 53}, \city{Brescia}, \postcode{ I-25123},  \country{Italy}}}

\abstract{Linear error-correcting codes can be used for constructing secret sharing schemes; however,finding in general the access structures of these secret sharing schemes  and, in particular, determining  efficient access structures is difficult.
Here we investigate the properties of  certain algebraic hypersurfaces over finite fields, whose intersection numbers with any hyperplane only takes a few values; these varieties give rise to $q$-divisible  linear codes with at most $5$ weights. Furthermore, for $q$ odd, these codes turn out to be minimal and we characterize the access structures of the secret sharing schemes based on their dual codes. Indeed, the secret sharing schemes thus obtained are democratic, that is each participant belongs to the same number of minimal access sets and can easily be described.}



\maketitle

\section{Introduction}
\label{sec:1}
A \emph{secret sharing scheme} (henceforth SSS for short) is a cryptographic technique for the management of the access to a
secret $s$  by a collective partnership.
The partners,  also called \emph{participants}, (to the scheme)  hold shares of information and  access is allowed only to certain qualified groups of them,
who can be given authorization by combining together their shares.
For any given SSS, a set of participants who can reconstruct the secret value $s$ from its shares is said to be an \emph{access set}.
Also, an access set is  \emph{minimal} if none of its proper subsets is in turn an access set itself.
The family of all the minimal access sets for a SSS is called the \emph{access structure} of the SSS.
In \cite{M92}, Massey devised a SSS based on linear codes and pointed out the relationship between
the minimal codewords of the dual code and the access structure.
 We refer to his work and \cite{DY03} for details
 on the construction and performance of these schemes.
 This provides a motivation
 for determining the set of all minimal codewords
 of an arbitrary linear code over a finite field. This problem turns out to be,
 in general, hard;
 however some useful  criteria might be obtained in the case of projective codes.
 Recently, several authors have considered this problem; see  e.g. \cite{DY03,YD06,LXL} and the references therein.

 In the present paper, we
 study certain algebraic hypersurfaces  in the  finite projective space $\PG(r,q^2)$ over  $\GF(q^2)$
 of dimension $r\geq 3$,
 whose intersection numbers with any hyperplane only take a few values.
 These hypersurfaces arise in the construction of quasi-Hermitian
 varieties and are, as such, also of independent interest; see~\cite{ACK}.
 Then we determine the  five weights of the corresponding   $q^2$-ary projective codes   proving that, except for $r=3$ and $q$ odd, these codes are all \emph{$q$-divisible}, that is their weights are  divisible by $q$; see \cite{W}.
 Finally, for $q$ odd, we show that these projective codes are minimal and  hence the related SSS's have an efficient access structures as they are democratic SSS's, namely, each participant is involved in the same number of minimal access sets. We also discuss how these access structures are
 related to the involved geometry.

The paper is organized into 6 sections.
Section~\ref{sec:2}  introduces the necessary background on quasi-Hermitian varieties and minimal codes. In Section~\ref{sec:3}, we exhibit an infinite family of $q^2$-ary minimal codes arising from quasi-Hermitian varieties.
In Section~\ref{sec:4}, we study the intersections of certain algebraic hypersurfaces of degree $2q$ over $\GF(q^2)$ with the extremal subspaces of $\PG(r,q^2)$ and,  as a byproduct,  in Section \ref{sec:5} we provide an infinite family of $q$-divisible  minimal codes.
In~\ref{sec:6} we consider the access structures arising from projective
codes and apply our results to the construction of infinite families of
SSS's,  which can be described algebraically.

Our main results are contained in Theorems \ref{fin}, \ref{thm:52} and \ref{thm:64}.

\section{Preliminaries}
\label{sec:2}
An $[n,r+1]_q$ \emph{projective system} is a collection  $\cV$ of $n$ not necessarily distinct points in the  projective space $\PG(r,q)$ over the finite field $\GF(q)$ of order $q$, with $q$ a prime power.
Fix a reference frame in $\PG(r, q)$,  with
homogeneous coordinates $(X_0,X_1,\dots, X_r)$,  and construct a matrix $G$ by taking as columns
the coordinates of the points of $\cV$,  suitably normalized. The code $\cC(\cV)$ having $G$ as generator
matrix is called \emph{the code determined from} $\cV$.

It is straightforward to see that, even if $\cC(\cV)$ is not
uniquely determined by $\cV$, all the codes that might be obtained
in this way are in fact equivalent; so we
shall often speak of \emph{the} projective code determined by $\cV$.

The spectrum of the intersections of $\cV$ with the hyperplanes of $\PG(r, q)$ is related with the list of the weights of
the associated code;
the {\em $k$-higher weights } of $\cC(\cV)$ are given by
\begin{multline*}
  d_k(\cC)=n-\max\{\mid \cV \cap \pi \mid : \pi \ \text{
      is a projective subspace of codimension} \\
    k \text{ in }\ \PG(r,q)\};
      \end{multline*}
 note  that the first higher weight $d_1(\cC(\cV))$ is actually the minimum distance of the code $\cC(\cV)$. We refer to \cite{TVN} for further details on this geometric approach to codes.

Error correcting codes can be used in order to devise \emph{access schemes}
or  SSS's.
In his  seminal work~\cite{M92}, Massey proposed the use of minimal
codewords in a \emph{dual code}, in order to specify the access structure
of a SSS.

\begin{definition}
  Let $\cC$ be a code of length $n$. For any codeword $c\in\cC$ and $1\leq i\leq n$
  the \emph{support} of $c$ is the set $\mathrm{supp}(c):=\{ i : c_i\neq 0 \}$
  of the positions of its non-zero components.
  Given $c,c'\in\cC$ we write that $c'\preccurlyeq c$ if
  $\mathrm{supp}(c')\subseteq\mathrm{supp}(c)$. We say that
  $c$ is a \emph{minimal word} of $\cC$ if $c'\preccurlyeq c$ implies
  that there is $\alpha\in\GF(q)$ such that
  $c'=\alpha c$.
\end{definition}
\begin{remark}
  According to our definition, a minimal codeword is just a word of $\cC$ which is
  not covered by any other linearly independent  codewords.
  From the point of view of
  the geometric description of the code, this is most convenient and
  this is consistent with the terminology of~\cite{LXL}.
  We point out that in~\cite{M92} for a codeword to be minimal
  it is also required that the leftmost non-zero
  component of the word must be $1$ (such codewords are called
  \emph{minimal AS-codewords} in \cite{LXL}).
  Accordingly, if $\cC$ is a linear code,  the minimal AS-codewords
  are exactly those minimal codewords of $\cC$ which lie in
  the affine space $\PG(\cC)\setminus \Sigma_{\infty}$, where $\Sigma_{\infty}$ is the hyperplane at infinity of equation $X_0=0$.
\end{remark}

It is well known that a linear code is spanned by its
minimal words and that all its minimum weight codewords
are minimal (according to our definition),
but little can be said about the remaining words in general.

\begin{definition}
 A linear code is a \emph{minimal  code}
 if all of its codewords are minimal.
\end{definition}
Ashikhmin and Barg in~\cite{AB} provided a sufficient
condition so that an $[n,k]$-linear
code $\cC$ over $\GF(q)$ is minimal, namely
\begin{equation}\label{suff}
  \frac{w_{\min}}{w_{\max}}>\frac{q-1}{q}, \end{equation}
where $w_{\min}$ and $w_{\max}$ are respectively the minimum and
maximum weight of non-zero codewords of $\cC$.

 Nice examples of linear codes can be obtained by considering some algebraic varieties of $\PG(r,q)$. In general, pointsets  with few intersection numbers with respect to the hyperplanes provide codes with interesting structure.

 Here we shall take into account
 \emph{quasi-Hermitian varieties} in $\PG(r,q^2)$,
 that is,  point sets  having the same  size and the same intersection numbers with respect to hyperplanes as a  non-singular Hermitian variety $\cH(r,q^2)$ of $\PG(r,q^2)$.
More in detail, the size of a quasi-Hermitian variety $\cV$  is \[\mid \cH(r,q^2)\mid =\frac{(q^{r+1}+(-1)^r)(q^{r}-(-1)^r)}{(q^2-1)},\]
whereas the intersection numbers of $\cV$ with respect to hyperplanes are \[\mid \cH(r-1,q^2)\mid =\frac{(q^r+(-1)^{r-1})(q^{r-1}-(-1)^{r-1})}{q^2-1},\]
 and
\[\mid \Pi_0 \cH(r-2,q^2)\mid =\frac{(q^r+(-1)^{r-1})(q^{r-1}-(-1)^{r-1})}{q^2-1}+(-1)^{r-1}q^{r-1},\]
 where  $\Pi_i$ is an $i$-dimensional space of $\PG(r,q^2)$  and $\Pi_0 \cH(r-2,q^2)$ is a cone, the join of the  vertex $\Pi_0$ to a non-singular Hermitian variety $\cH(r-2,q^2)$ of a projective subspace $\Pi_{r-2}$ which does not contain  $\Pi_0$.
A non-singular Hermitian variety of $\PG(r,q^2)$ is, by definition, a quasi-Hermitian variety: the \emph{classical quasi-Hermitian variety}.

\section{Projective minimal codes}
\label{sec:3}
In recent years \cite{BBG,BB21,BLT,TC} constructions of minimal
codes arising from projective systems have been considered.
In particular,  in \cite{TC} it has been proved that minimal projective codes
are equivalent to projective systems which are
\emph{cutting} $1$-blocking sets (see~\cite{BB21} for the definition).

Recall that a cutting $1$-blocking set (or, in brief a \emph{cutting} blocking
set) $\cH$ is a subset of $\PG(r,q)$ such that for any hyperplane
$\Sigma$ we have $\langle \Sigma\cap\cH\rangle=\Sigma$.

Let $\cH$ be a set of points of $\PG(r,q)$,  and $\cC(\cH)$ be one of its associated projective codes.
It is straightforward to see that the words of
a projective code determined by $\cH\subseteq\PG(V)$, where
$\langle\cH\rangle=\PG(V)$  and $V$ is the vector space underlying $\PG(r,q)$, correspond exactly to the evaluations of the
elements of the dual $V^*$ of $V$ on the given projective system.
The following result was proved in~\cite{ABN} and independently in~\cite{TC}.

\begin{theorem}
  \label{mwcw}
  Let $\cH$ be a set of $N$ points  of $\PG(r,q)$
  such that $\langle\cH\rangle=\PG(r,q)$.
  For each $i \in \{ 1 \dots N \}$  let $P_i\in\GF(q)^{r+1}$ be
  a fixed representative of a point $[P_i]\in\cH$ and denote by
  $\cC(\cH)$  the projective linear code having generator matrix whose columns
  are the vectors $P_i$. Then $\cC (\cH)$ is a minimal code
   if and only if for any hyperplane $\Sigma$ of $\PG(r,q)$
  \begin{equation}\label{cut} \langle \Sigma\cap\cH\rangle=\Sigma.
  \end{equation}

\end{theorem}
A general problem is to determine when the set of points of an algebraic
variety $\cH$ in $\PG(r,q)$ is a cutting blocking set.
It is easy to see that elliptic quadrics in $\PG(3,q)$, as well as degenerate hypersurfaces,
in general, are not.
In \cite{BLT} the authors proved that if $\cH$ is a non-singular Hermitian variety in a given canonical form or a
quadric in  projective dimension $r\geq 4$, then $\cH$ is a  cutting blocking set and thus $\cC(\cH)$ is minimal.
We can easily extend the same result to any quasi-Hermitian variety of $\PG(r,q^2)$.
If $\cH$ is a quasi-Hermitian variety of $\PG(r,q^2)$ then $\cC(\cH)$ has $2$ weights. It is  straightforward to see that
Condition \eqref{suff} does not hold but Theorem~\ref{mwcw} can be applied as follows.

\begin{theorem}\label{herm}
  Let $\cH$ be a quasi-Hermitian variety in $\PG(r,q^2)$. Then,
  $\cC(\cH)$ is a minimal code.
\end{theorem}
\begin{proof}
  Since a quasi-Hermitian variety is a projective variety whose
  intersections with hyperplanes have the same cardinalities as
  the intersection of a Hermitian variety, we first see that
  $\cH$ cannot be contained in any hyperplane of $\PG(r,q^2)$.
  Also, suppose that there is a hyperplane $\varphi$ such
  that $\dim(\langle\cH\cap \varphi\rangle)<r-1$.
  This means that $\cH\cap\varphi\subseteq\Sigma$ with
  $\Sigma$ a projective space of dimension at most $r-2$.
  Thus, we would have
  \[ \mid \cH\cap \varphi\mid \leq\frac{q^{2r-2}-1}{q^2-1}, \]
  which is not possible. The thesis now follows from Theorem~\ref{mwcw}.

\end{proof}

\begin{remark}

  If we want to study codes arising from higher degree functions defined
over some algebraic varieties $\cH$, a convenient setting is to use Veronese embeddings
and represent these codes (in turn) as projective codes.
In particular, to investigate quadratic sections of $\cH$,  we just
apply the quadratic Veronese embedding
\[ \nu_r^2:\begin{cases}
    \PG(r,q)\to\PG(\frac{r^2+3r}{2},q) \\
    [(x_0,\dots,x_r)]\to [(x_0^2,x_0x_1,\dots,x_0x_r,x_1^2,\dots,x_r^2)]
  \end{cases} \]
and then consider Theorem~\ref{mwcw}.\\
We denote this new code,  arising from the projective system
of $\nu_r^2(\cH)$ as $\cC^2$.
Observe however, that in general it is not true that $\langle \nu_r^2(\cH) \rangle=\PG(\frac{r^2+3r}{2},q)$.
\end{remark}

\section{Hypersurfaces with  few intersection numbers}
\label{sec:4}
 In $\PG(r,q^2)$ with
homogeneous coordinates $(X_0,X_1,\dots, X_r)$, consider the affine
space $\AG(r,q^2)$ whose infinite hyperplane $\Sigma_{\infty}$ has
equation $X_0=0$. Then $\AG(r,q^2)$ has affine coordinates $(x_1,x_2,\dots,x_r)$
where $x_i=X_i/X_0$ for $i\in \{1,\dots,r\}$.
Consider the algebraic variety  $\cB$ of affine equation
 \begin{equation}\label{eqg}
x_r^q-x_r+\alpha^q(x_1^{2q}+\dots+x_{r-1}^{2q})-\alpha(x_1^2+\dots+x_{r-1}^2)=(\beta^q-\beta)(x_1^{q+1}+\dots+x_{r-1}^{q+1}),
\end{equation}
where $\alpha \in GF(q^2)^*$, $\beta\in GF(q^2)\setminus GF(q)$ and  the following conditions  are satisfied:
for odd $q$,
\begin{itemize}
\item[\rm(1)] $r$ is odd and $4\alpha^{q+1}+(\beta^q-\beta)^2 \neq 0$, or
\item [\rm(2)]$r$ is even and  $4\alpha^{q+1}+(\beta^q-\beta)^2$ is a non--square in $\GF(q)$;
\end{itemize}
 for  even $q>2$,
 \begin{itemize}
 \item [\rm(i)] $r$ is odd, or
\item[\rm(ii)]
 $r$ is even and $ \Tr(\alpha^{q+1}/(\beta^q+\beta)^2)=0$.
\end{itemize}

In~\cite{ACK} the authors proved that gluing together  the set of affine points of
$\cB$ with the degenerate Hermitian  variety \[ \cF=\{(0,X_1,\dots,X_r)\colon X_1^{q+1}+\dots+X_{r-1}^{q+1}=0\} \]
of the hyperplane $X_0=0$
gives a quasi-Hermitian variety $\cH=(\cB \cap AG(r,q^2))\cup \cF$.
In this section we study the intersection numbers of
$\cB$ with hyperplanes by
a ``surgery'' argument, that is
\begin{itemize}
\item first we consider the intersection of a hyperplane $\Sigma$ with
  the quasi-Hermitian variety $\cH$;
\item then we remove from this intersection its points at infinity  which are contained in $\cF$
  and add the possible intersections of $\Sigma$ with $\cB_{\infty}:=\cB \cap \Sigma_{\infty}$.
\end{itemize}
More concisely, our arguments are based on  these facts
\begin{enumerate}

\item $ \Sigma\cap\cB =((\Sigma\cap\cH)\setminus\cH_{\infty}) \cup  (\Sigma\cap\cB_{\infty})$;
\item  the cardinalities of $\Sigma \cap \cH$ are known;
\item  $\mid \cB \cap AG(r,q^2)\mid =\mid \cH \setminus \cH_{\infty}\mid =q^{2r-1}$.
\end{enumerate}
In particular $\mid \Sigma\cap\cB\mid =\mid \Sigma\cap\cH\mid -\mid \Sigma\cap\cH_{\infty}\mid +\mid \Sigma\cap\cB_{\infty}\mid $.
Observe that for $r=2$ $\cF=\{P_{\infty}\}$ and $\cB=\cH$ is a Buekenhout-Metz unital of $\PG(2,q^2)$,
see~\cite{GE}.
\subsection{Case  $r\geq 3$ and $q$ odd }
We compute the intersection numbers of $\cB$ with respect to any  hyperplane $\Sigma$ of $\PG(r,q^2)$, with $q$ odd.
The intersection between the algebraic variety  $\cB$ and $\Sigma_{\infty}$ is the degenerate quadric  $\cB_{\infty}$ of $\Sigma_{\infty}$ with equation \[x_1^2+\dots+x_{r-1}^2=0.\]
In this section $\cP_i$ will denote a parabolic quadric of an $i$-dimensional projective space $\Pi_i$ for $i$ even, whereas  $\cI_i$ and $\cE_i$ will denote a hyperbolic and an elliptic quadric of $\Pi_i$,  with $i$ odd. We set $\cI_{-1}=\emptyset$ and $\cH(-1,q^2)=\cH(0,q^2)=\emptyset$.

Assume that $ r\geq 3$ is odd.
In this case $\cB_{\infty}$ is a cone with vertex the point $P_{\infty}(0,0,\dots, 0,1)$ and basis a hyperbolic quadric $\cI_{r-2}$ of an $(r-2)$-dimensional projective space contained in  $\Sigma_{\infty}$,  hence
   $ \mid \cB_{\infty}\mid =q^2[(q^{2(r-2)}-1)/(q^2-1)+q^{r-3}]+1$ and  $\mid \cB\mid =q^{2r-1}+q^2[(q^{2(r-2)}-1)/(q^2-1)+q^{r-3}]+1.$

Let $\Sigma$ be a hyperplane passing through the point $P_{\infty}$. Then $\Sigma$ meets $\cB_{\infty}$ in a cone with vertex the point
$P_{\infty}$ and basis either a parabolic quadric   $\cP_{r-3}$ or a cone $\Pi_0 \cI_{r-4}$.
Now,  suppose that  the hyperplane $\Sigma$ does not contain $P_{\infty}$. We observe that $\Sigma$ meets $\cB_{\infty}$ in a hyperbolic quadric $\cI_{r-2}$ of an $(r-2)$-dimensional projective space contained in $\Sigma_{\infty}$.
Thus, the possible values of $\mid \Sigma\cap\cB\mid $, for any hyperplane $\Sigma $ of $\PG(r,q^2)$,
are:
\vskip.1cm
\begin{enumerate}[(C1)]
\item $\mid \cB_{\infty}\mid $;
\item $\mid \cH(r-1,q^2)\mid -\mid P_{\infty}\cH(r-3,q^2)\mid +\mid P_{\infty}\cP_{r-3}\mid $;
\item $\mid \cH(r-1,q^2)\mid -\mid P_{\infty}\cH(r-3,q^2)\mid +q^2\mid \Pi_0 \cI_{r-4}\mid +1$;
\item $\mid \Pi_0\cH(r-2,q^2)\mid -\mid P_{\infty}(\Pi_0\cH(r-4,q^2))\mid +q^2\mid \Pi_0 \cI_{r-4}\mid +1$;
\item  $\mid \Pi_0\cH(r-2,q^2)\mid -\mid P_{\infty}(\Pi_0\cH(r-4,q^2))\mid +\mid P_{\infty}\cP_{r-3}\mid $;
\item $\mid \cH(r-1,q^2)\mid -\mid \cH(r-2,q^2)\mid +\mid \cI_{r-2}\mid $;
\item $\mid \Pi_0\cH(r-2,q^2)\mid -\mid \cH(r-2,q^2)\mid +\mid \cI_{r-2}\mid $.
\end{enumerate}

In increasing order we get the following intersection numbers $n_i$ with $i=1,\dots, 5$:
  \newcommand{\onA}{q^2\frac{(q^{2(r-2)}-1)}{q^2-1}+q^{r-1}+1}
  \newcommand{\onB}{q^{2r-3}+ q^{r-1}+\frac{q^{2(r-2)}-q^2}{q^2-1}+1}
  \newcommand{\onC}{q^{2r-3}-q^{r-2}+q^{r-3}+\frac{q^{2(r-2)}-1}{q^2-1}}
  \newcommand{\onD}{q^{2r-3}+\frac{q^{2(r-2)}-q^2}{q^2-1}+1}
  \newcommand{\onE}{q^{2r-3}+q^{r-1}-q^{r-2}+q^{r-3}+\frac{q^{2(r-2)}-1}{q^2-1}}
  \begin{itemize}
\item (C1) gives $n_1=\onA$;
\item (C6) gives  $n_2=\onC$;
\item (C2) and (C5) yield $n_3=\onD$;
\item (C7) provides $n_4=\onE$;
\item (C3) and (C4) provide
  $n_5=\onB$.
\end{itemize}

Now suppose that  $r\geq 4$ is even. In this case  $\cB_{\infty}$ is a cone with vertex the point $P_{\infty}(0,0,\dots, 0,1)$ and basis a parabolic quadric $\cP_{r-2}$ in an $r-2$-dimensional projective space  contained in  $\Sigma_{\infty}$ and it contains $q^2[(q^{2(r-2)}-1)/(q^2-1)]+1$ points over $\GF(q^2)$.
As $\cB \cap AG(r,q^2)$ contains $q^{2r-1}$ affine points, we get
\[\mid \cB\mid =q^{2r-1}+q^2[(q^{2(r-2)}-1)/(q^2-1)]+1.\]
We observe that a generic hyperplane of $\Sigma$ which does not pass through $P_{\infty}$ meets $\cB_{\infty}$ in a parabolic quadric $\cP_{r-2}$ of an $r-2$-dimensional projective space in $\Sigma_{\infty}$.
On the other hand, if $\Sigma$ contains $P_{\infty}$ then it meets $\cB_{\infty}$ in a cone with vertex $P_{\infty}$ and basis either a cone $\Pi_0'\cP_{r-4}$, or a hyperbolic  quadric $\cI_{r-3}$, or an elliptic quadric $\cE_{r-3}$.
Thus,  the possible values of $\mid \Sigma\cap\cB\mid $, for any hyperplane $\Sigma $ of $\PG(r,q^2)$, are:
\vskip.1cm
\begin{enumerate}[(C1)]
\item $\mid \cB_{\infty}\mid $;
\item $\mid \cH(r-1,q^2)\mid -\mid P_{\infty}\cH(r-3,q^2)\mid +q^2\mid \Pi_0'\cP_{r-4}\mid +1$;
\item $\mid \cH(r-1,q^2)\mid -\mid P_{\infty}\cH(r-3,q^2)\mid +q^2\mid \cE_{r-3}\mid +1$;
\item $\mid \cH(r-1,q^2)\mid -\mid P_{\infty}\cH(r-3,q^2)\mid +q^2\mid \cI_{r-3}\mid +1$;
\item $\mid \Pi_0\cH(r-2,q^2)\mid -\mid P_{\infty}(\Pi_0\cH(r-4,q^2))\mid +q^2\mid \Pi_0'\cP_{r-4}\mid +1$;
\item $\mid \Pi_0\cH(r-2,q^2)\mid -\mid P_{\infty}(\Pi_0\cH(r-4,q^2))\mid +q^2\mid \cE_{r-3}\mid +1$;
\item $\mid \Pi_0\cH(r-2,q^2)\mid -\mid P_{\infty}(\Pi_0\cH(r-4,q^2))\mid +q^2\mid \cI_{r-3}\mid +1$;
\item $\mid \cH(r-1,q^2)\mid -\mid \cH(r-2,q^2)\mid +\mid \cP_{r-2}\mid $;
\item $\mid \Pi_0\cH(r-2,q^2)\mid -\mid \cH(r-2,q^2)\mid +\mid \cP_{r-2}\mid $.
\end{enumerate}

In increasing order,  we  obtain as the possible  intersection numbers  of $\cB$ with respect to the hyperplanes the following
  \newcommand{\enA}{q^2\frac{(q^{2(r-2)}-1)}{q^2-1}+1}
  \newcommand{\enB}{q^{2r-3}+\frac{(q^{2(r-2)}-q^2)}{q^2-1}+1}
  \newcommand{\enC}{q^{2r-3}-q^{r-1}+q^{r-2}+\frac{q^{2(r-2)}-1}{q^2-1}}
  \newcommand{\enD}{q^{2r-3}+\frac{q^{2(r-2)}-q^2}{q^2-1}-q^{r-2}+1}
  \newcommand{\enE}{q^{2r-3}+\frac{q^{2(r-2)}-q^2}{q^2-1}+q^{r-2}+1}
\begin{itemize}
\item (C1) gives $n_1=\enA$;
\item (C9) gives  $n_2=\enC$;
\item (C3) and (C6) yield $n_3=\enD$;
\item (C2) and (C5) provide $n_4=\enB$;
\item (C4), (C7) and (C8) provide $n_5=\enE$.
\end{itemize}
We summarize our results in the following theorem.

\begin{theorem}\label{qodd}
Suppose $q$ to be an odd prime power. Then the hypersurface
$\cB$ of $\PG(r,q^2)$, $r\geq 3$,  with equation \eqref{eqg} contains $q^{2r-1}+q^{r-1}+(q^{2(r-1)}-q^2)/(q^2-1)+1$  points if $r$ is odd or
$q^{2r-1}+(q^{2(r-1)}-q^2)/(q^2-1)+1$ points  if $r$ is even. Furthermore its possible intersection sizes with  hyperplanes are:
\begin{itemize}
\item for $r$ odd:
\[n_1= \onA,\qquad n_2= \onC, \]
\[n_3= \onD,\qquad n_4= \onE, \]
\[n_5= \onB; \]
\item for $r$ even:
\[ n_1= \enA,\qquad
n_2= \enC, \]
\[n_3= \enD,\qquad
n_4= \enB, \]
\[n_5= \enE. \]
\end{itemize}
\end{theorem}

\subsection{Case $r\geq 3$  and $q$ even}

In this case, the intersection between the algebraic variety  $\cB$ and $\Sigma_{\infty}$ is the degenerate quadric  $\cB_{\infty}$ which consists of the single hyperplane of $\Sigma_{\infty}$:  $x_1+\dots+x_{r-1}=0$. Therefore
  the size of $\cB$ is:
  \[ q^{2r-1}+q^{2(r-2)}+\dots+q^2+1,\]
  and the possible  intersection numbers of $\cB$ with respect to hyperplanes of $\PG(r,q^2)$ are:
  \vskip.1cm
\begin{enumerate}[(C1)]
\item $\mid \cB_{\infty}\mid $;
\item $\mid \cH(r-1,q^2)\mid -\mid P_{\infty}\cH(r-3,q^2)\mid +\mid \Pi_{r-2}\mid $;
\item $\mid \cH(r-1,q^2)\mid -\mid P_{\infty}\cH(r-3,q^2)\mid +\mid \Pi_{r-3}\mid $;
\item $\mid \Pi_0\cH(r-2,q^2)\mid -\mid P_{\infty}(\Pi_0\cH(r-4,q^2)\mid +\mid \Pi_{r-2}\mid $;
\item $\mid \Pi_0\cH(r-2,q^2)\mid -\mid P_{\infty}(\Pi_0\cH(r-4,q^2)\mid +\mid \Pi_{r-3}\mid $;
\item $\mid \cH(r-1,q^2)\mid -\mid \cH(r-2,q^2)\mid +\mid \Pi_{r-3}\mid $;
\item $\mid \Pi_0\cH(r-2,q^2)\mid -\mid \cH(r-2,q^2)\mid +\mid \Pi_{r-3}\mid $.
\end{enumerate}
\vskip.1cm

So, for $r$ odd we obtain the following weights:
\begin{itemize}
\item (C1) gives $n_1=\frac{q^{2(r-1)}-1}{q^2-1}$;
\item (C6) gives  $n_2=q^{2r-3}-q^{r-2}+\frac{q^{2(r-2)}-1}{q^2-1}$;
\item (C3) and (C5) provide $n_3=q^{2r-3}+\frac{q^{2(r-2)}-1}{q^2-1}$;
\item (C7)  yields $n_4=q^{2r-3}+q^{r-1}-q^{r-2}+\frac{q^{2(r-2)}-1}{q^2-1}$;
\item (C2) and (C4) provide $n_5= q^{2r-3}+\frac{q^{2(r-1)}-1}{q^2-1}$.
\end{itemize}
For $r$ even we obtain
\begin{itemize}
\item (C1) gives $n_1=\frac{q^{2(r-1)}-1}{q^2-1}$;
\item (C7) gives  $n_2=q^{2r-3}-q^{r-1}+q^{r-2}+\frac{q^{2(r-2)}-1}{q^2-1}$;
\item (C3) and (C5) provide $n_3=q^{2r-3}+\frac{q^{2(r-2)}-1}{q^2-1}$;
\item (C6)  yields $n_4=q^{2r-3}+q^{r-2}+\frac{q^{2(r-2)}-1}{q^2-1}$;
\item (C2) and (C4) provide $n_5= q^{2r-3}+\frac{q^{2(r-1)}-1}{q^2-1}$.
\end{itemize}

 \begin{theorem}\label{qeven}
Suppose $q$ to be even and $r\geq 3$. Then the hypersurface $\cB$ of Equation~\eqref{eqg} has $q^{2r-1}+q^{2(r-2)}+\dots+q^2+1$ points in $\PG(r,q^2)$ and the following intersection sizes with respect to hyperplanes:
\begin{itemize}
\item for $r$ odd:
  \[n_1=\frac{q^{2(r-1)}-1}{q^2-1}, \qquad
    n_2=q^{2r-3}-q^{r-2}+\frac{q^{2(r-2)}-1}{q^2-1},\]
\[ n_3=q^{2r-3}+\frac{q^{2(r-2)}-1}{q^2-1}, \qquad
n_4=q^{2r-3}+q^{r-1}-q^{r-2}+\frac{q^{2(r-2)}-1}{q^2-1}, \]
\[n_5=q^{2r-3}+\frac{q^{2(r-1)}-1}{q^2-1};\]
\item for $r$ even:
\[n_1=\frac{q^{2(r-1)}-1}{q^2-1}, \qquad n_2=q^{2r-3}-q^{r-1}+q^{r-2}+\frac{q^{2(r-2)}-1}{q^2-1}, \]
\[n_3=q^{2r-3}+\frac{q^{2(r-2)}-1}{q^2-1},\qquad
  n_4=q^{2r-3}+q^{r-2}+\frac{q^{2(r-2)}-1}{q^2-1},\]
\[n_5=q^{2r-3}+\frac{q^{2(r-1)}-1}{q^2-1}.\]
\end{itemize}
\end{theorem}

\subsection{Line sections of $\cB$ in $\PG(r,q^2)$}

Our aim  is to provide the spectrum of all possible intersection numbers
between $\cB$ and   a line   of $\PG(r,q^2)$. We are going to prove the following theorem.
\begin{theorem}\label{lines}
Let $\ell$ be a line of $\PG(r, q^2)$. Then, the possible sizes for $\ell \cap \cB$ are as follows
\[0,1,2,q-1, q,q+1,q+2, 2q-1,2q,q^2+1. \]
\end{theorem}

\begin{proof} Let us assume $q$ to be odd and  consider a line $\ell$  of $\PG(r,q^2)$. If $\ell$ is contained in $\Sigma_{\infty}$ then the possible  sizes of $\ell \cap \cB$ are $0$, $1$,$2$ or $q^2+1$.  Now suppose that $\ell \nsubseteq \Sigma_{\infty}$   and  $\mid \ell \cap \cB \mid \geq 1 $.
  From~\cite{ACK}, it can be directly
  seen that the collineation group
  of $\cB$ acts transitively on its affine points.
  Thus, we can assume that $\ell$ passes through the origin of the fixed reference system. We have to study the following system
 % \begin{small}
\begin{equation}
    \left\{\begin{array}{l}\label{sis1}
     x_r^q-x_r+\alpha^q(x_1^{2q}+\dots+x_{r-1}^{2q})-\alpha(x_1^2+\dots+x_{r-1}^2)=\\
       (\beta^q-\beta)(x_1^{q+1}+\dots+x_{r-1}^{q+1}),\\
      x_1=m_1t\\
      x_2=m_2t\\
     \vdots\\
      x_r=m_rt,
      \end{array}\right.
  \end{equation}
% \end{small}
\noindent where $t$ ranges over $\GF(q^2)$.\\
First we consider the case in which $m_r\neq 0$ and hence we can assume $m_r=1$.
In order to study System~\eqref{sis1},
choose a primitive element $\gamma$ of $\GF(q^2)$
and let $\varepsilon=\gamma^{(q+1)/2}$.
We now regard $\GF(q^2)$ as a vector space over $\GF(q)$ with a fixed
basis $\{1,\varepsilon\}$
and write the elements in $\GF(q^2)$ as linear
combinations with respect to this basis, that is,
$x_i=x_i^{(0)}+x_i^{(1)}\varepsilon$.
Then,
$\varepsilon^q=-\varepsilon$ and $\varepsilon^2$ is a primitive
element of $\GF(q)$. With this choice of $\varepsilon$,
setting
\[ u=m_1^{(0)} m_1^{(1)}+\dots+m_{r-1}^{(0)}m_{r-1}^{(1)}, \] \[ v=\left(m_1^{(0)}\right)^2+\dots+ \left(m_{r-1}^{(0)}\right)^2\quad \text{ and }
\quad  z=\left(m_1^{(1)}\right)^2+\dots+ \left(m_{r-1}^{(1)}\right)^2. \]
Equation~\eqref{sis1} gives
%\begin{small}
\begin{multline}
\label{eqodd1}
[2\alpha_0u+\alpha_1(\varepsilon^2z+v)+\beta_1(\varepsilon^2 z-v)]t_0^2+
\varepsilon^2 [2\alpha_0u+\alpha_1(v+\varepsilon^2z)+\beta_1(v-\varepsilon^2z)]t_1^2+\\
2[\alpha_0(\varepsilon^2z+v)+2\alpha_1\varepsilon^2u]t_0t_1+t_1=0
\end{multline}
%\end{small}
The solutions $(t_0,t_1)$ of
\eqref{eqodd1}  can be viewed as the affine points of the (possibly degenerate)
 conic
$\Gamma$ of $\PG(2,q)$
associated to the symmetric $3\times 3$ matrix
\begin{small}
\[A=\begin{pmatrix}
      2\alpha_0u+\alpha_1(\varepsilon^2z+v)+\beta_1(\varepsilon^2 z-v) &\alpha_0(\varepsilon^2z+v)+2\alpha_1\varepsilon^2u  & 0\\

 \alpha_0(\varepsilon^2z+v)+2\alpha_1\varepsilon^2u  & \varepsilon^2[2\alpha_0u+\alpha_1(v+\varepsilon^2z)+\beta_1(v-\varepsilon^2z)] & 1/2 \\
  0 &  1/2 & 0\\
\end{pmatrix}. \]
\end{small}
The number of affine points of $\Gamma$ equals
the number of points in
$\AG(3,q^2)$ which lie in $\cB\cap\ell$.
Observe that $\mathrm{rank}(A)\geq  2$.
Let us first suppose $\det(A) \neq 0$. In this case $\Gamma$ is a non-degenerate conic in $\PG(2,q)$ and hence has either $q-1$ or $q$  or $q+1$  affine points.
If $\mathrm{rank}(A)=2$ then $\Gamma$ is the union of two distinct lines  either defined over $\GF(q)$ or defined over $\GF(q^2)$ and conjugate to each
other. This means that the number of  affine points of $\Gamma$ is  either $2q-1$ or $2q$ or $1$.
Thus if $\ell \cap \cB_{\infty}=\emptyset$ then $\mid \ell \cap \cB\mid \in \{1,q-1,q,q+1,2q-1,2q\}$.

Now suppose that $\ell$ meets $\mathcal{B}_{\infty}$.
The point at infinity of $\ell$ is $R=(0,m_1,m_2 \dots, m_r)$ and $R$ is also a point of $\cB_{\infty}$  if and only if
\[
 m_1^2+\dots+m_{r-1}^2=0,
\]
that is,  $\sum_{i=1}^{r-1}(m_i^{(0)}+\varepsilon m_i^{(1)})^2=0$.
This can be rewritten as
\[ \sum_{i=1}^{r-1}(m_i^{(0)})^{2}+\varepsilon^2(m_i^{(1)})^2 + 2 \varepsilon \sum_{i=1}^{r-1}m_i^{(0)}m_i^{(1)}=0, \] and hence we get
\[ \begin{cases}\displaystyle
    \sum_{i=1}^{r-1}(m_i^{(0)})^{2}+\varepsilon^2 (m_i^{(1)})^2=0 \\
    \displaystyle
    \sum_{i=1}^{r-1}m_i^{(0)}m_i^{(1)}=0.
\end{cases} \]
Thus  if $R \in \cB_{\infty}$ then  $v+\varepsilon^2z=0$ and  $u=0$. In this case $A$ becomes
 \[A=\begin{pmatrix}
  \beta_1(\varepsilon^2 z-v) &0  & 0\\
 0 & \varepsilon^2\beta_1(v-\varepsilon^2z) & 1/2 \\
  0 &  1/2 & 0\\
 \end{pmatrix}. \]
If $\det(A)\neq 0$ then $\Gamma$ is an ellipse as $\varepsilon^2$ is a non--square of $\GF(q)$.
In the case in which $\mathrm{rank}(A)=2$,  then we get $u=v=z=0$ and $\Gamma$ consists of $q$ affine points. Thus,  $\mid \ell \cap \cB\mid  \in \{q+1, q+2\}$.

Now  suppose that $m_r= 0$. In this case  the number of points in
$\AG(3,q^2)$ which lie in $\cB\cap\ell$ equals the number of affine points of the degenerate conic $\Gamma$ with associated matrix
\begin{small}
 \[A=\begin{pmatrix}
  2\alpha_0u+\alpha_1(\varepsilon^2z+v)+\beta_1(\varepsilon^2 z-v) &\alpha_0(\varepsilon^2z+v)+2\alpha_1\varepsilon^2u  & 0\\
 \alpha_0(\varepsilon^2z+v)+2\alpha_1\varepsilon^2u  & \varepsilon^2[2\alpha_0u+\alpha_1(v+\varepsilon^2z)+\beta_1(v-\varepsilon^2z)] & 0 \\
  0 &  0 & 0\\
\end{pmatrix}. \]
\end{small}
If $\mathrm{rank}(A)=2$ then $\Gamma$ has either $1$ or $2q-1$ or $2q$ points. Otherwise, $\mathrm{rank}(A)=1$ and $\Gamma$ consists of $q$ points,  or the matrix $A$ is the null matrix, namely  $\Gamma$ is the entire affine plane and $\ell \subset \cB$.
Furthermore, in the case in which $\mid \ell \cap \cB_{\infty}\mid  =1$,  it is easy to see that  $\Gamma \cap \AG(2,q^2)$  consists of a single point over $\GF(q)$ or it is the entire plane.
Hence  $\mid \ell \cap \cB\mid \in \{1,2,q,2q-1, 2q,  q^2+1\}$ and our theorem follows for $q$ odd.

Let us consider the case of $q$ even. As $q>2$, we can fix a basis for $\GF(q^2)$ over $\GF(q)$ as $\{1,\epsilon\}$, with $\epsilon\in \GF(q^2)\setminus \GF(q)$  such that
$\epsilon^2+\epsilon +\nu =0$, for some $\nu \in \GF(q)\setminus \{1\}$,
 with $\Tr_q(\nu)=1$.
Then $\epsilon^{2q}+\epsilon^q+\nu =0$ and hence
$(\epsilon^q+\epsilon)^2+(\epsilon^q + \epsilon)=0$, leading to $\epsilon^q+\epsilon +1 =0$.
With this choice of $\varepsilon$, setting as before $u=m_1^{(0)} m_1^{(1)}+\dots+m_{r-1}^{(0)} m_{r-1}^{(1)}$, $v=(m_1^{(0)})^2+\dots+ (m_{r-1}^{(0)})^2$ and $z=(m_1^{(1)})^2+\dots+ (m_{r-1}^{(1)})^2$,  \eqref{sis1}
gives
\begin{multline}
\left[\beta_1 (u+v+\nu z)+\alpha_1 (v+z+\nu z) +\alpha_0 z\right]t_0^2\\
+\left[(\beta_1 \nu (u+v+\nu z)+ \alpha_1 \nu (v+z+\nu z) +(\alpha_1 +\alpha_0) (v+z)\right]t_1^2 \\
+ \beta_1 (u+v+\nu z)t_0t_1 +t_1=0.
\end{multline}
which can be viewed  again as the equation of a conic $\Gamma $ of $\AG(2,q^2)$.

It is straightforward to see that $\mid \Gamma \cap \AG(2,q^2)\mid \in \{1,q-1, q,q+1, 2q-1,2q, q^2 \}$.
Arguing as in the $q$ odd case,  the proof is completed.

\end{proof}

\section{Codes with $5$ weights}
\label{sec:5}
We are going to
determine the parameters of the projective code generated from the hypersurface $\cB$ of Equation~\eqref{eqg} and in particular its weight enumerator for $r=3$ and $q$ odd.

\begin{theorem}\label{fin}Let $q$ be an odd prime power. Then, the points of $\cB$ in $\PG(r,q^2)$, $r>3$ determine a $q$-divisible minimal projective code $\cC(\cB)$  of length  $N=q^{2r-1}+q^{r-1}+(q^{2(r-1)}-q^2)/(q^2-1)+1$ for $r$ odd, or
  $N =q^{2r-1}+(q^{2(r-1)}-q^2)/(q^2-1)+1 $ for $r$ even,  dimension $r+1$
  and  non-zero weights:
  \begin{itemize}
\item for  $r$ odd:
  \[ w_5=q^{2r-1}-q^{2r-3}+q^{2(r-2)},\]
  \[
  w_4=q^{2r-1}-q^{2r-3}+q^{2(r-2)}+q^{r-2}-q^{r-3}, \]
  \[ w_3=q^{2r-1}-q^{2r-3}+q^{2(r-2)}+q^{r-1}, \]
  \[ w_2=q^{2r-1}-q^{2r-3}+q^{2(r-2)}+q^{r-1}+q^{r-2}-q^{r-3},
    \qquad w_1=q^{2r-1}; \]
\item for $r$ even:
\[
  w_5=q^{2r-1}-q^{2r-3}+q^{2(r-2)}-q^{r-2},\qquad
 w_4=q^{2r-1}-q^{2r-3}+q^{2(r-2)}, \]
\[ w_3=q^{2r-1}-q^{2r-3}+q^{2(r-2)}+q^{r-2}, \]
\vskip-.1cm
\[ w_2=q^{2r-1}-q^{2r-3}+q^{2(r-2)}+q^{r-1}-q^{r-2},
  \qquad
w_1=q^{2r-1}. \]
\end{itemize}
\end{theorem}
\begin{proof}
  Since $w_i=N-n_i$ where the $n_i$'s are the intersection numbers of $\cB$ with respect to the hyperplanes of $\PG(r,q^2)$,
from Theorem \eqref{qodd} we have just to prove that $\cC(\cB)$ is a minimal code. We restrict ourselves to the case  $r$  odd. Under this hypothesis
the maximal weight  of $\cC(\cB)$ is $w_1=q^{2r-1}$ whereas the minimal one is $w_5=q^{2r-1}-q^{2r-3}+q^{2(r-2)}$. We observe that
$\frac{w_5}{w_1}>\frac{q^2-1}{q^2}$,  that is,  Condition~\eqref{suff}  is satisfied and  hence $C(\cB)$ is a minimal code.
\end{proof}

From Theorem~\eqref{qeven} we obtain the following.
\begin{theorem}
\label{thm:52}
  Let $q$ be an even prime power.
Then, the points of $\cB$ in $PG(r,q^2)$, $r\geq 3$ determine a  $q$-divisible projective code $\cC(\cB)$  of length  $N=q^{2r-1}+q^{2(r-2)}+q^{2(r-3)}+\dots+q^2+1$, dimension $r+1$
and non-zero weights:
\begin{itemize}
\item for $r$ odd:
  \[
    w_5= q^{2r-1}-q^{2r-3},\qquad
    w_4= q^{2r-1}-q^{2r-3}+q^{2r-4}-q^{r-1}+q^{r-2}, \]
\[ w_3=q^{2r-1}-q^{2r-3}+q^{2(r-2)}, \]
\[
w_2=q^{2r-1}-q^{2r-3}+q^{2(r-2)}+q^{r-2},\qquad
  w_1=q^{2r-1}; \]
\item for $r$ even:
\[ w_5= q^{2r-1}-q^{2r-3},\qquad
  w_4= q^{2r-1}-q^{2r-3}+q^{2(r-2)}-q^{r-2}, \]
\[ w_3=q^{2r-1}-q^{2r-3}+q^{2(r-2)}, \]
\[
w_2= q^{2r-1}-q^{2r-3}+q^{2(r-2)}+q^{r-1}-q^{r-2},\qquad
w_1=q^{2r-1}. \]
\end{itemize}
\end{theorem}
\begin{remark}
  For $q$ even, the code $\cC(\cB)$ is not a minimal code,
  as the support of the words of weight $w_5$ is contained
  in the support of  words of weight $w_1$.
 This is consistent with Theorem~\ref{mwcw},
 as $\langle\cB_{\infty}\rangle=\cB_{\infty}\neq\Sigma_{\infty}$ in this
 case. So, the $q^2-1$ words of weight
 $w_1=q^{2r-1}$ are exactly those which are not
 minimal.
\end{remark}
Let $A_j$ denote  the number
of codewords of $\cC(\cB)$ of weight $j$.

\begin{proposition}\label{pro3}
The points of $\cB$ in $\PG(3,q^2)$, with $q$ an odd prime power,  determine a minimal  projective code $\cC (\cB)$ of length $N=q^5+2q^2+1$,
non-zero weights:
   \[
     w_1 =q^5,\quad
     w_2 = q^5-q^3+2q^2+q-1,\quad
     w_3 = q^5-q^3+2q^2, \]
   \[
     w_4 = q^5-q^3+q^2+q-1,\quad
     w_5 = q^5-q^3+q^2,
   \]
   and weight enumerator
   $w(x) := \sum_{i} A_i x^i$, where\\

 $A_0=1$, \ $A_{w_1}=q^2-1$, \ $A_{w_2} =(q^6-q^5+q^3)(q^2-1)$, \ $A_{w_3}=(q^4 - q^2)(q^2-1)$, \
    \[A_{w_4}=(q^5-q^3)(q^2-1),  \ A_{w_5}=2q^2(q^2-1)\]
and all of the remaining $A_i$'s are $0$.
\end{proposition}
\begin{proof}
As $w_i=N-n_i$,  where $n_i$'s are the intersection numbers of $\cB$ with respect to the planes, the first part of our theorem follows from  Theorem~\eqref{qodd} for $r=3$ and from the fact that $w_5/w_1>(q^2-1)/q^2$.  We are going to compute the weight enumerator of the code. We observe that $\Sigma_{\infty}$ is the only hyperplane meeting $\cB$ in $2q^2+1$ points and this means  $A_{w_1}=q^2-1$. Also, through the point $P_{\infty}=(0,0,0,1)$ there pass $2q^2$ planes meeting $\Sigma_{\infty}\cap \cB$ in one line and $q^4-q^2$ planes meeting  $\Sigma_{\infty}\cap \cB$ just at the point $P_{\infty}$. Hence, $A_{w_5}=2q^2(q^2-1)$ and $A_{w_3} =(q^4-q^2)(q^2-1)$.

Finally, we recall that $\cH= (\cB \cap AG(3,q^2))\cup \cF$  is a quasi-Hermitian variety of $\PG(3,q^2)$.  Let us call a  plane intersecting $\cH$ in $i$ points an  $(i)$-$plane$ of $\cH$. Using  the following properties of $\cH$:
 \begin{itemize}
 \item the number of $(q^3+q^2+1)$-planes is $q^5+q^2+1$ ,
 \item the number of $(q^3+1)$-planes is  $q^6-q^5+q^4$,
 \item  $P_{\infty}$ lies on $q^4-q^3$  $(q^3+1)$-planes and on $q^3+q^2+1$ $(q^3+q^2+1)$-planes of $\cH$,
 \end{itemize}
   we get $A_{w_4}=(q^5 - q^3)(q^2-1)$ and  $A_{w_2}=(q^6-q^5+q^3)(q^2-1)$.
\end{proof}

 \begin{remark}
Theorems~\ref{lines} and~\ref{fin} yield that the higher weights $d_1(\cC(\cB))$ and $d_{r-1}(\cC(\cB))$ are as follows:
\begin{itemize}
\item for $q$ odd:
\begin{itemize}
\item if $r$ is odd:\\
\quad
  $ d_1(\cC(\cB)) =q^{2r-1}-q^{2r-3}+q^{2(r-2)}$,\\
   \quad   $d_{r-1}(\cC(\cB))= q^{2r-1}+q^2\frac{(q^{2(r-2)}-1)}{q^2-1}+q^{r-1}-q^2;$
\item if $r$ is even:\\
\quad
  $ d_1(\cC(\cB)) =q^{2r-1}-q^{2r-3}+q^{2(r-2)}-q^{r-2}$,\\
\quad  $d_{r-1}(\cC(\cB))= q^{2r-1}+q^2\frac{(q^{2(r-2)}-1)}{q^2-1}-q^2$;
\end{itemize}
\item for $q$ even and any $r$:\\
 \quad $d_1(\cC(\cB)) = q^{2r-1}-q^{2r-3},$\\
\quad $d_{r-1}(\cC(\cB))=q^{2r-1}+q^{2(r-2)}+\dots+q^4$.
\end{itemize}

We leave to a future work to determine  the higher weights  $d_k(\cC(\cB))$ for all $1<k<r-1$.
\end{remark}

\section{Secret sharing schemes from hypersurfaces}
\label{sec:6}
In this section we recall a method for constructing SSS's based on linear codes and then we present a class of SSS's using the hypersurfaces introduced in Section~\ref{sec:4}.

Let $C$ be an $[n,k,d]_q$-linear code.
In the SSS $\fS(C)$ based on $C$, the secret is an element of $\GF(q)$, and $n-1$ parties
$P_1,P_2,\dots,P_{n-1}$ as well as a trusted third party are involved. To compute the shares with respect to a secret $s$, the trusted third party randomly chooses a vector ${\bf{u}}=(u_0,\dots, u_{k-1})\in \GF(q)^k$ such that $s = {\bf{u}} g_0$ and $G=(g_0, g_1,\dots g_{n-1})$ is a generator matrix of $C$. There are altogether $q^{k-1}$ such vectors ${\mathbf{u}}  \in \GF(q)^k$. The third
party then treats $\bf{u}$ as an information vector and computes the corresponding codeword
$t=( t_0,t_1,\dots,t_{n-1})= {\bf{u}}G$.
He then gives $t_i$ to party $P_i$ as share for each $i\geq1$.
Note that $t_0 = {\bf{u}}g_0 = s$. It is easily seen that a set of shares $\{t_{i_1},t_{i_2} ,\dots, t_{i_m}\}$  determines the secret if and only if $g_0$ is a linear combination of $g_{i_1},\dots, g_{i_m}$ .
The following properties hold.

\begin{proposition}\cite{M92}
Let $G$ be a generator matrix of an $[n, k; q]$-code $C$. In the SSS based on $C$, a set of shares
$\{t_{i_1},t_{i_2} ,\dots, t_{i_m}\}$ determines the secret if and only if there exists a codeword
\[{\bf{c}}=(1, 0, \dots,0, c_{i_1}, 0,\dots, 0, c_{i_m}, 0,\dots,0 )\]
in the dual code $C^{\perp}$ where $c_{i_j}\neq 0$ for at least one $j$, $1\leq i_1<\dots  < i_m\leq n-1$ and $1 \leq m \leq n-1$.
If there is a codeword like $\bf{c}$ in $C^{\perp}$, then the vector
\[g_0=\sum_{j=1,\dots, m} x_j g_{i_j} \]
where $x_j \in \GF(q)$ for $1 \leq j \leq  m$. Then the secret $s$ is recovered by computing
\[s=\sum_{j=1,\dots, m} x_j t_{i_j}.\]
\end{proposition}

If a set of participants can recover the secret by combining their shares, then any group of participants containing this set can also recover the secret.

\begin{definition}
  A set of participants is called a \emph{minimal access set} if they can recover the secret by combining their shares and none of its
  proper subsets can do so. Here,  a proper subset has fewer members than this set.
  The \emph{access structure} $\fA(\cC)$ of the SSS
  $\fS(\cC)$ is the set of its minimal access sets.
\end{definition}

\begin{proposition}\label{con}\cite{M92}
Let $C$ be an $[n, k; q]$-code, and let $G = (g_0,g_1,\dots, g_{n-1})$ be its generator matrix. If each nonzero codeword of $C$ is a minimal word, then in the SSS based on $C^{\perp}$, there are altogether $q^{k-1}$ minimal access sets. In addition, we have
the following:
\begin{enumerate}
\item If $g_i$ is a multiple of $g_0$, $1 \leq i \leq n-1$, then participant $P_i$ must be in every minimal access set. Such a participant is called a \emph{dictatorial} participant.
\item If $g_i$ is not a multiple of $g_0$, $1 \leq i \leq n-1$,  then participant $P_i$ must be in $(q-1)q^{k-2}$ out of $q^{k-1}$ minimal access sets.
\end{enumerate}
\end{proposition}
We refer the reader to~\cite{M92} for the actual construction of the
SSS.

In this section we shall consider the access structures of SSS's arising from codes constructed from hypersurfaces. These access
structures turn out to reflect the geometry of the hypersurface and they also
afford a compact description in terms of their automorphism groups.

\begin{proposition}
  Let $C$ and $C'$ be two equivalent $[n,k,d]$-codes over $\GF(q)$ with
  generator matrices $G$ and $G'$ with $G'=RGPD$ where $R$ is a $k\times k$
  invertible matrix, $P$ is an $n\times n$ permutation matrix and $D$ is
  an invertible $n\times n$ diagonal matrix.
  Suppose that
  the permutation $\sigma:\{0,\dots,n-1\}\to\{0,\dots,n-1\}$,  induced by $P$,
  fixes $0$.
  Then there is a bijection between the shares of the SSS's
  $\fS(C)$ and $\fS(C')$ as well as between the corresponding access
  structures $\fA(C)$ and $\fA(C')$.
\end{proposition}
In light of the above proposition, given an hypersurface $\cV$ of $\PG(V)$
and a fixed point $P_0\in\cV$, we can construct many equivalent $[n,k,d]$-linear
codes with a generator matrix having $P_0$ as its first column.

So, we propose the following notation.
Let $\cV$ be an hypersurface, and let $P_1$ be a chosen point of $\cV$;
  let $C=\cC(\cV;P_0)$ be a projective code arising from $\cV$,  with a generator matrix having $P_0$ as its first column.
 We denote the SSS's  based on $C$ by the symbol $\fS(\cV;P_0)$ and
  the SSS's based on $C^{\perp}$ by $\fS(\cV^{\perp};P_0)$.

\begin{remark}
  Suppose $\cV\subseteq\PG(V)$.
  The elements of the access structure $\fA(\cV^{\perp};P_0)$ correspond
  to the support of the subsets $(\cV\setminus\{P_0\})\setminus\Pi$ of
  $\cV$ as $\Pi$ varies among the hyperplanes of $\PG(V)$ not containing
  the point $P_0$. In particular,  we can describe $\fA(\cV^{\perp};P_0)$ directly, without explicit mention of
  the projective code $C^{\perp}$ induced by $\cV$.
\end{remark}

\begin{definition}
  We say that two access structures $\fA$ and $\fA'$ associated to
 SSS's $\fS$ and $\fS'$ with set of participants
 $X$ and $X'$, respectively, are \emph{equivalent} if there is a
 bijection $\theta:X\to X'$ such that
 $\fA'=\{ S^{\theta}\colon S\in\fA\}$.
\end{definition}

\begin{definition}
  Let $\fS$ be a SSS with corresponding access structure $\fA$ and
  set of participants $X$. We say that $\gamma\in\mathrm{Sym}(X)$ is an
  \emph{automorphism} of  $\fA$ if,  for any $T\in\fA$,  we have
  $T^{\gamma}:=\{ \gamma(t) \colon t\in T\}\in\fA$.
  Given a subgroup $\Gamma$ of $\mathrm{Sym}(X)$ and some elements
  $S_1,\dots,S_t\in\fA$, we say that $\fA$ is a
  \emph{$\Gamma$-development} of the \emph{starters} $S_1,\dots,S_t$ if
  \[ \fA=\{ S_i^{\gamma}\colon \gamma\in\Gamma, i\in 1,\dots,t \}. \]
\end{definition}

\begin{proposition}
  Let $\cV=\{P_0,\dots,P_{n-1}\}$ be an algebraic variety of $\PG(V)$.
  Denote by $\fA(\cV^{\perp};P_0)$
  the access structure of the SSS  $\fS(\cV;P_0)$.
  Let  $\gamma$ be a collineation of $\PG(V)$ fixing $P_0$ and
  such that $\gamma(\cV)=\cV$. Then $\hat{\gamma}\in\mathrm{Sym}(1,\dots,n-1)$, where
  \[ \hat{\gamma}(i)=j \Leftrightarrow \gamma(P_i)=P_j \]
  acts as an
  automorphism of $\fA(\cV^{\perp};P_0)$.
\end{proposition}
\begin{proof}
  The elements of $\fA(\cV^{\perp};P_0)$ correspond to the
  support of the complement of the intersection of $\cV$ with the hyperplanes
  $\Sigma$ of $\PG(V)$ for which $P_0\not\in\Sigma$.
  Since $\gamma$ preserves $\cV$ and fixes $P_0$, it also acts on
  the hyperplanes of $\PG(V)$ not through $P_0$; as such $\hat{\gamma}$
  acts on the access structures as a permutation group.
\end{proof}
\begin{remark}
  Let $\cV$ be a projective hypersurface in $\PG(V)$ and suppose
  that $\cV$ admits a transitive group of automorphisms. Then, for any
  $P,Q\in\cV$ we have that $\fA(\cV^{\perp};P)$ is equivalent to
  $\fA(\cV^{\perp};Q)$.
\end{remark}

We now apply these notions to quasi-Hermitian varieties.
Let $\cB$ be the hypersurface of Equation~\eqref{eqg} and let $P_0$ be
a fixed point of $\cB$. Denote by
$C=\cC(\cB)$  an associated projective $q^2$-ary $[N,r+1,d_1]$-code.

Using Proposition~\ref{con}, we shall first prove that for $q$ odd,  the SSS $\fS(\cB^{\perp};P_0)$ based on  the dual code of $\cC(\cB)$ is \emph{democratic}, that is,   each participant is involved in the same number of minimal access sets, no matter the choice of the point $P_0$.

\begin{theorem}
\label{thm:64}
Let $r\geq 3$ and $q$ be an odd prime power. In the SSS
$\fS(\cB^{\perp};P_0):=\fS(C^{\perp})$
based on the dual code $C^{\perp}:=\cC (\cB)^{\perp}$, there are altogether $q^{2r}$ minimal access sets and
  $m=N-1$ participants where
\begin{itemize}
\item
$m =q^{2r-1}+q^{r-1}+\frac{(q^{2(r-1)}-q^2)}{q^2-1}$, for $r$ odd;
\item
$m =q^{2r-1}+\frac{(q^{2(r-1)}-q^2)}{q^2-1}$, for $r$ even.
\end{itemize}
Furthermore, each participant $P_i$,  $\forall i=1\dots m$, is involved in exactly $(q^2-1)q^{2(r-1)}$ out of $q^{2r}$ minimal access sets.
\end{theorem}
\begin{proof}
Let $G$ be a generator matrix of $\cC(B)$ and let $g_i$ denote the $i$-column of $G$. We observe that $\forall i\neq j$ $g_i$ is not a multiple of $g_j$. Thus, the result follows from Theorem~\ref{fin} and Propositions~\ref{pro3} and \ref{con}.
\end{proof}

\begin{remark}
  We observe that if $\cH$ is a quasi-Hermitian variety of $\PG(r,q^2)$, where $r\geq 2$ and $q$ any prime power, then the projective code $\cC(\cH)$ is a two-weight code.   From Theorem~\ref{herm}  we also know that  $\cC(\cH)$ is minimal. In particular, Property 2 of Proposition~\ref{con} applies and the SSS based on the dual of $\cC(\cH)$ turns out to be  democratic. For $q$ odd, the SSS associated to the aforementioned  code has the same number of minimal access sets as the SSS based on $\cC(\cB)^{\perp}$    but it has a different number of  participants, that is $\frac{(q^{r+1}+(-1)^r)(q^{r}-(-1)^r)}{(q^2-1)}-1$ for $r>2$.
\end{remark}

We now present a detailed example of access structure with a rich automorphism
group by considering the case of Hermitian varieties.

\begin{example}\label{Hermitiana}
  Let us consider the Hermitian surface $\mathcal{H}$ of $\PG(3,q^2)$.
  Then, the projective code $\cC(\cH)$ has parameters $[(q^3+1)(q^2+1),4,q^5]$
  and its weight distribution is given by $A_0=1$, $A_{q^5}=(q^4-1)(q^3+1)$
  and $A_{q^5+q^2}=q^8-q^7-q^4+q^3+1$. Also the automorphism group
  $\mathrm{PGU}(4,q)$ of
  $\mathcal{H}$ is transitive on $\mathcal{H}$.
  So, no matter what point $P$ is chosen in $\mathcal{H}$, all access structures
   $\fA(\mathcal{H}^{\perp};P)$ are equivalent.
  Since $\cC(\mathcal{H}^{\perp})$ is a
  $[(q^3+1)(q^2+1),(q^3+1)(q^2+1)-4,3]$ code, in the SSS $\fS(\mathcal{H}^{\perp};P)$
  there are $q^6$ minimal access sets
  and each participant is involved in $q^6-q^4$ of them.
  As seen before, each minimal access set $A\in\fA$ corresponds to a plane $\pi_A$
  not through a fixed point $P_1$ of $\cH$ and, consequently, it has size
  $\mid \cH\setminus\pi_A\mid -1$. In particular, since $\cH$ is a two-intersection
  set with respect to the (hyper)planes, the possible sizes of the minimal access
  sets are $q^5-1$ and $q^5+q^2-1$. Observing that there is
  exactly one plane through $P_1$ meeting $\cH$ in $q^3+q^2+1$ points, while the
  remaining $q^4+q^2$ planes meet $\cH$ in $q^3+1$ points, it is straightforward
  to determine how many access structures of each type there are from the weight enumerator of $\cC(\cH)$.

  We now describe in further detail the structures for $q=2$.
  In this case the size of $\fA$ is $64$ and we can easily see that
  $32$ of its minimal access sets  have size $31$ and  the remaining $32$ are of  size $35$.
  The stabilizer $\Gamma$ of a point $P_1$ in $\mathrm{PGU}(4,2)$
  has size $576$. It has
  $5$ orbits, say $\Omega_i$ (with $i=1,\dots,5$),
  on the hyperplanes of $\PG(3,4)$ of size respectively
  $1, 8, 12, 32$, and $32$, respectively. The union of the first $3$ orbits is the set
  of the hyperplanes through the fixed point $P_1$; the orbits
  $\Omega_4$ and $\Omega_5$ correspond to the families of hyperplanes
  with intersection, respectively, $9$ and $13$ with $\cH$. In turn
  these correspond to the access structures of size $35$ and $31$.

  Denote by $\{P_1,\dots,P_{45}\}$ the points of $\cH$.
  It can be seen that
  $\Gamma$ acts on the $P_i$'s as the permutation group generated by
  \begin{multline*}
    \gamma_1:= (2,44,38,3,42,40)(4,10,28,7,13,30)(5,22,15,8,24,12)
    (6,34,20,9,32,18) \\
    (11,19,29,14,23,33)(16,25,21)(26,31,27)(37,41,45)(39,43),
  \end{multline*}
  \begin{multline*}
    \gamma_2:=(10,12,11)(13,15,14)(16,36,26)(17,45,31)(18,43,32)(19,44,30)
    \\ (20,39,34)(21,37,35)(22,38,33)(23,42,28)(24,40,29)(25,41,27),
  \end{multline*}
  \begin{multline*}
    \gamma_3:=  (4,12)(5,10)(6,11)(7,15)(8,13)(9,14)(16,26)(17,35)(18,33)(19,34)\\
    (20,29)(21,27)(22,28)(23,32)(24,30)(25,31).
  \end{multline*}
  By the previous remarks $\Gamma$
  acts in a natural way on the minimal access sets of
  $\fA$. In particular, $\fA$ is the $\Gamma$-development of the
  following two starters:
  \begin{multline*}
    S_1:=\{2, 3, 4, 5, 6, 7, 8, 13, 14, 15, 16, 17, 18, 19, 20, 21, 22, 26,
    \\ 30, 31,
    32, 33, 34, 35, 36, 37, 38, 39, 43, 44, 45 \},
\end{multline*}
  \begin{multline*}
    S_2:=\{
  2, 3, 4, 5, 6, 7, 8, 10, 11, 12, 13, 14, 17, 18, 19, 20, 21, 22, 23, 24, \\
  25, 26, 27, 28, 29, 33, 34, 35, 36, 40, 41, 42, 43, 44, 45
  \}.
  \end{multline*}
\end{example}
\begin{remark}
  We point out that all access structures arising from Hermitian varieties
  are the $\Gamma$-development of $2$ elements under the action of a
  group $\Gamma$ which is isomorphic to the stabilizer of an isotropic
  point in $\mathrm{PGU}(r+1,q)$.
\end{remark}

\section*{Thanks}
The authors wish to thank M. Bonini and M. Borello for having pointed out
several very recent references on the relationship between minimal linear
codes and projective systems, as
well as for having corrected a misattribution present on the first version of the
preprint.
\section*{Acknowledgements}

This research was carried out within the activities of the GNSAGA - Gruppo Nazionale per le Strutture Algebriche, Geometriche e le loro Applicazioni of the Italian INdAM.

 \end{document}